%% file: main.tex
\documentclass[a4paper,UKenglish]{article}

\usepackage{microtype}
\usepackage{a4wide}
\usepackage{units}
\usepackage{amssymb,amsthm}
\usepackage{mathtools}
\usepackage{dsfont}
\usepackage{graphicx}
\usepackage{epsfig}
\usepackage{wrapfig}
\usepackage[ruled,vlined]{algorithm2e}
\usepackage{enumerate}
\usepackage[shortlabels]{enumitem}
\usepackage{upgreek}
\usepackage{todonotes}
\usepackage{url}
\usepackage{thrmappendix}
\usepackage{refcount}
\usepackage{amsmath}
\usepackage{enumerate}

\providecommand{\newmu}{\mu'}

\providecommand{\Ib}{I_b}
\providecommand{\MH}{M_H}
\providecommand{\NH}{N_H}
\providecommand{\IC}{I_C}
\providecommand{\ID}{I_D}
\providecommand{\Eps}{\eps}
\providecommand{\devfromone}{\eta}

\makeatletter
\renewcommand{\paragraph}{%
	\@startsection{paragraph}{4}%
	{\z@}{1.2ex \@plus .5ex \@minus .1ex}{-.7em}%
	{\normalfont\normalsize\bfseries}%
}
\makeatother

\usepackage{tikz}
\usetikzlibrary{decorations.pathreplacing}
\usetikzlibrary{plotmarks}
\usetikzlibrary{positioning,automata,arrows}
\PassOptionsToPackage{usenames,dvipsnames,svgnames}{xcolor}

\title{An Integer Interior Point Method for Min-Cost Flow Using Arc Contractions and Deletions}
\author{Ruben Becker$^{1,2,3}$ and 
Andreas Karrenbauer$^{1,2}$ and Kurt Mehlhorn $^{1}$\\[1mm]
  \normalsize{\texttt{firstname.lastname@mpi-inf.mpg.de}}\\[2mm]
  \normalsize{$^{1}$ Max Planck Institute for Informatics, Saarbr\"ucken, Germany} \\[1mm]
  \normalsize{$^{2}$ Max Planck Center for Visual Computing and Communication} \\[1mm]
  \normalsize{$^{3}$ Saarbr\"ucken Graduate School for Computer Science}
}
\date{}

\input{macros.tex}

\newtheorem{theorem}{Theorem} 
\newtheorem{lemma}{Lemma}
\newtheorem{definition}{Definition}

\begin{document}

\maketitle

\begin{abstract}
	We present an interior point method for the min-cost flow problem that uses arc contractions and deletions to steer clear from the boundary of the polytope when path-following methods come
	too close. We obtain a randomized algorithm running in expected
	$\tilde O( m^{3/2} )$ time that only visits integer lattice points in the
	vicinity of the central path of the polytope. This enables us to use integer arithmetic like
	classical combinatorial algorithms typically do. We provide explicit bounds on
	the size of the numbers that appear during all computations. By presenting an
	integer arithmetic interior point algorithm we
	avoid the tediousness of floating point error analysis and achieve a method that
	is guaranteed to be free of any numerical issues. We thereby eliminate one of the drawbacks of numerical methods in contrast to combinatorial min-cost flow algorithms that still yield the most efficient implementations in practice, despite their inferior worst-case time complexity.
\end{abstract}

\section{Introduction}\label{sec:intro}
The min-cost flow problem is one of the most general network flow problems
that can be solved in polynomial time. It generalizes the max-flow as well as
the shortest path problem. Moreover, it has many applications, since a lot of
real-world problems can be modeled as a min-cost flow problem. Non-surprisingly
there is a rich literature concerning the problem. The work can roughly be split
into two major lines. First, there is the line of combinatorial algorithms that starts with the seminal work of Ford and
Fulkerson~\cite{Ford:2010:FN:1942094} and continues with Edmonds and
Karp~\cite{DBLP:conf/aussois/EdmondsK01}, who gave the first polynomial time
algorithm for the problem. Further polynomial time algorithms are the strongly polynomial time algorithm of
Orlin~\cite{DBLP:conf/stoc/Orlin88} and the scaling
approaches~\cite{DBLP:conf/stoc/GoldbergT87,Goldberg:1990:FMC:92217.92225}
that peak in the double scaling technique by Ahuja et
al.~\cite{DBLP:journals/mp/AhujaGOT92} achieving a bound of
$O(nm\log \log U\cdot \log(nC))$, where $n$ denotes the number of nodes,
$m$ the number of arcs, $U$ the maximal capacity, and $C$ the maximal cost in the network.
The second line of work originates in the results of
Vaidya~\cite{DBLP:conf/focs/Vaidya89a} who was the first to specialize
interior point methods to network-flow problems. Lagging behind the
combinatorial methods at that time, this line of research has become extremely
fruitful in the last decade;
most prominently with the work of Daitch and
Spielman~\cite{DBLP:conf/stoc/DaitchS08} who gave a randomized
$\tilde O(m^{3/2})$-time algorithm.\footnote{ The tilde indicates that
logarithmic factors in $n,U,C$ are hidden.} They were the first to use
nearly-linear time equation
solvers~\cite{DBLP:conf/stoc/SpielmanT04,DBLP:conf/focs/KoutisMP10,DBLP:conf/stoc/KelnerOSZ13} in the context of generalized min-cost flow problems for approximately solving systems of linear equations in each of $\tilde O(\sqrt{m})$ iterations, thereby, for the first time, achieving sub-quadratic bounds in sparse graphs and under the similarity assumption, i.e., when the input data is polynomially bounded in $n$. The same running time (up to log factors) was
achieved in~\cite{DBLP:conf/isaac/BeckerK14} for the classical min-cost flow problem with a
simple potential reduction algorithm. Even more recently,
the breakthrough result of Lee and Sidford~\cite{DBLP:conf/focs/LeeS14}
for solving general linear programs in rank-many iterations yields a very
intricate, randomized, interior point method of complexity $\tilde O(\sqrt{n}m)$ for the min-cost flow problem. Most recently, Cohen et al.~developed a new framework to analyze interior point methods~\cite{DBLP:journals/corr/CohenMSV16}. They obtain a bound of $\tilde O(m^{3/2})$ for the case of min-cost flow problems with the restriction that there are only unit capacities.

Although achieving the currently best-known bounds in theory, the interior point methods lack
behind the combinatorial methods in practice, see for example~\cite{DBLP:conf/dimacs/GoldbergK91,DBLP:journals/jal/Goldberg97,DBLP:conf/alenex/BeckerFK16} for efficient implementations of combinatorial algorithms and~\cite{DBLP:journals/corr/abs-1207-6381} for an extensive comparison of implementations.
We are not aware of any interior point method that was reported to be competitive with
the combinatorial approaches on real world problems. The reasons are many.
First, combinatorial algorithms are rather easy to implement and necessary subroutines are readily available in standard libraries, whereas the ingredients that are required to achieve the better theoretical results are rather involved and use complicated
subroutines such as low-stretch spanning trees~\cite{DBLP:conf/stoc/AbrahamN12} or spectral vertex sparsifiers~\cite{DBLP:journals/corr/LeePS15} that are
already difficult to implement efficiently on their own.
Second, interior point methods are formulated for real arithmetic and implementations require
high-precision floating point numbers in order to avoid numerical
instabilities, even if all input data is integer. In contrast to this, the combinatorial algorithms have the advantage that,
given an integral input, they get along with fixed-precision integer arithmetic,
which makes them very robust in practice. Note that, if floating-point-arithmetic is used, even combinatorial algorithms may suffer from numerical instability as was observed by Althaus and Mehlhorn in LEDA~\cite{DBLP:journals/ipl/AlthausM98}.
In this paper, we address the second issue of practical implementations of interior point methods decribed above by giving an interior point method
that uses integer arithmetic on numbers of polynomial length.
We proceed by defining the min-cost flow problem formally.

An instance of the min-cost flow problem is defined by a directed
graph $G=(V,A)$ with $|V|=n$ and $|A|=m$, node demands $b \in \ZZ^n$ with $\ones^T b = 0$,
arc costs $c\in \ZZ^m$ and arc capacities $u \in\NN^m$.
We define $\beta:=\gcd(b,u)$ and $\gamma:=\gcd(c)$, i.e.\ greatest common divisors of the entries of the respective vectors. Moreover, we use $U:=\max\{\|u\|_{\infty}/\beta,\|b\|_{1}/(2\beta)\}$ and $C:=\|c\|_{\infty}/\gamma$.
A flow $x\in \RR^m$ is called \emph{feasible}, if $0 \le x \le u$ and $x(\din(v)) - x(\dout(v)) = b_v$
for every $ v\in V$. \footnote{We write $f(S):= \sum_{a\in S} f_a$ for $S\subseteq A$ for any vector $f\in \RR^m$ and for a node $v\in V$, we use $\din(v)$ for the ingoing arcs of $v$ and $\dout(v)$ for its outgoing arcs, respectively.}
The task is to find the feasible flow $x^*$ of minimal cost, i.e., $c^T x^* \le c^Tx$ for all feasible flows $x$, or to assert that no feasible flow exists.
It is well-known, see, e.g.,~\cite{DBLP:books/daglib/0069809,DBLP:conf/isaac/BeckerK14}, that we can restrict ourselves to the problem being feasible, having only non-negative costs, and no capacity constraints. We will also justify the last assumption in Appendix~\ref{sec:initial}.

Assuming the problem to be uncapacitated, we will use the notation $(G,b,c)$ for a min-cost flow instance on the graph $G$, with demand vector $b$ and costs $c$, which can be written as a primal-dual linear programming pair
\begin{align}
	\label{form:inputproblem}
	\min\{c^Tx \; : \; A x = b \text{ and } x \ge 0\} =  \max\{b^Ty \; : \; A^T y + s = c \text{ and } s \ge 0\},
\end{align}
where $A\in \{-1,0,1\}^{n\times m}$ denotes the node-arc incidence matrix of $G$, i.e., $A$ contains a column $\alpha$ for every arc $(v,w)\in A$ with $\alpha_v = -1$, $\alpha_w = 1$ and $\alpha_i=0$ for all $i\notin\{v,w\}$.

Path following methods maintain a pair of
primal/dual feasible solutions $x$ and $(y,s)$ such that $x_a s_a \approx \mu$
for all $a \in A$ for a positive parameter $\mu$, which they drive to zero.
They thereby follow the so-called \emph{central path}, which is the set of all points satisfying $x_as_a=\mu$ for all $a\in A$.
As a consequence of $\mu$ approaching zero, the \emph{duality gap} $c^Tx - b^T y = x^T s \approx \mu m$ goes to zero.
It may however happen, that either $x_a$ or $s_a$ becomes tiny while the other becomes huge.
Such behavior is numerically unpleasant as it requires to simultaneously deal with huge and tiny numbers.

\subsection{Our Contribution}
In this paper, we resolve the issue of simultaneously having to deal with huge and tiny numbers in interior point methods for min-cost flow by only considering the arcs with sufficiently large values of $x_a$ and $s_a$.
That is, whenever $x_a$ becomes too small, we delete the arc $a$ from the network and whenever $s_a$ becomes too small,
we contract the arc $a$. We thereby obtain a reduced problem in form of a minor of the original network.
We do not revoke any of these operations until termination of our algorithm and show that a near-optimal solution for the reduced problem can be lifted efficiently to an optimal solution of the original network.
Arc deletion and contraction has been used in combinatorial optimization before~\cite{DBLP:journals/combinatorica/Tardos85,orlin1984genuinely}, but, to our knowledge,
not in combination with path-following methods for network flow problems.

For the reduced problem, we can show that one can stay sufficiently close to the central path by hopping over points on an integer lattice in the vicinity of the central path.
More precisely, we show that if the greatest common divisor (gcd) of the demands and capacities and the gcd of the costs are sufficiently large, we can perform all arithmetic operations on integers as typical for the practically superior combinatorial algorithms. The requirement on the gcds can be easily achieved by an initial scaling of the demands, capacities and costs. Such scaling comes at a very low price because the expected running time of $\tilde O( m^{3/2} )$ only depends logarithmically on these scaling factors. We remark that it appears obvious that integer arithmetic can be achieved by scaling with a huge number that is, say, doubly exponential in $m,n$. However, we show that scaling with numbers that are polynomial in $n$ is sufficient. Moreover, we manage to make the exact numerical requirements of our algorithm explicit by stating absolute bounds on the appearing numbers, without hiding anything in $O$-notation.

\begin{theorem}\label{thm:contribution}
	Optimal primal and dual solutions for a min-cost flow problem
	can be computed in expected $\tilde O( m^{3/2} )$ time using only integer arithmetic on numbers bounded by
	$2^{31}m^{10}U^2C^2$, where $U$ and $C$ are bounds on the maximum capacity and cost, respectively.
\end{theorem}
Note that the running time of our algorithm matches the time bound of Daitch and Spielman~\cite{DBLP:conf/stoc/DaitchS08}, which dominates the known bounds for combinatorial algorithms in many settings.
For a better comparison, we remark that Daitch and Spielman~\cite{DBLP:conf/stoc/DaitchS08} give a bound on the bit length of the numbers appearing in their algorithms, as well. They state that the methods work on numbers of bit length at most $O(\log (n U/\eps))$, where $\eps=(12m^2U^3)^{-1}$ for solving the standard min-cost flow problem exactly \footnote{Note that in the case of generalized min-cost flow, the algorithms of Daitch and Spielman only give $\eps$-approximate solutions in general and thus the numerical requirements will also depend on the desired $\eps$. For the standard min-cost flow problem they show that an $\eps$ of the size as described above is sufficient.}. If the hidden constant in the $O$-notation is known, say, to be $k$, then their method could also be used with integer arithmetic on numbers bounded by $12^kn^km^{2k}U^{4k}$.
For robust implementations, it is advantageous to know the exact precision requirements, as an appropriate scaling can be performed in the beginning without having to worry about numerical issues or overflows in the course of the algorithm. Alternatively, one can still use adaptive precision, check the correctness of the output, and repeat the computation with higher precision, if necessary. However, this approach usually requires changes to the implementation as insufficient precision may lead to diverging behavior of sub-routines, see for example~\cite{DBLP:journals/comgeo/KettnerMPSY08}. Adding a timer is no solution either as it turns diverging behavior into worst-case behavior. However, such precautions are mandatory and the only choice, if the bounds are not known explicitly.

For the potential-reduction method from~\cite{DBLP:conf/isaac/BeckerK14}, we only gave a bound on the number of arithmetic operations and a strict bound on the size of the emerging numbers is non-trivial to obtain, if not even impossible. It is essential for our proof in this paper to exploit graph minors and to use a path-following method instead.
Furthermore, it is interesting to note that our lower bound on the value of primal variables does not depend on the costs, capacities, or demands in the network, but only on $n$ and $m$. This might be interesting in the context of strongly polynomial time algorithms for the capacitated min-cost flow problem, where the current record bound is $\tilde O(m^2)$, due to Orlin~\cite{DBLP:conf/stoc/Orlin88}.

We want to remark that we consider this work rather as a first important step in the direction of making numerical requirements of interior point methods for combinatorial problems apparent, than as giving the ultimate final answer to practitioners. We feel that our paper, despite the fact that it treats all numerical details, is still fairly elegant. This is due to the fact that we compute in integers, and by this, avoid the tediousness of floating point error analysis.

\paragraph*{Overview} At the beginning of Section~\ref{sec:pathfol}, we describe the high-level idea of the interior point method and explain how numerical issues come into play. In Subsection~\ref{subsec:minors}, we present our way of dealing with this problem, deletion and contraction of arcs, and treat some theoretical observations that allow us to pursue this way. We end this section with the concept of a so-called proxy. We then describe our algorithm in more detail and also give a pseudocode implementation of it, see Subsection~\ref{subsec:algorithm}. The interior points that the algorithm maintains are actually interior points of a slightly perturbed instance, we describe this in Subsection~\ref{subsec:perturbation}, where we will see that an optimal tree solution of the perturbed instance will be sufficient for obtaining an optimal tree solution for the original problem.
Subsection~\ref{subsec:crossover} describes how to obtain this tree solution, namely using the crossover algorithm from~\cite[Section 2]{DBLP:conf/isaac/BeckerK14}, see also~\cite{DBLP:conf/alenex/BeckerFK16} for more details on the crossover algorithm. We treat the centering step and the related numerical issues in Section~\ref{sec:centering step}. Finally, we conclude in Section~\ref{sec:conclusion}.

\section{Integer Interior Point Algorithm}\label{sec:pathfol}
Our method belongs to the class of path-following interior point methods. These methods maintain a pair of primal/dual feasible solutions $x,s \in \RR^m$ while more or less uniformly approaching the complementary slackness conditions. That is, the conditions are parameterized uniformly by a positive $\mu$, i.e., $x_a s_a = \mu$ for all $a \in A$, s.t.\ one obtains a pair of primal/dual optimal solutions in the limit $\mu \to 0$. It can be shown~\cite{DBLP:books/daglib/0094154} that the points satisfying theses conditions form a curve called \emph{central path}, which is contained in the interior of the primal/dual polyhedron. Moreover, moving along the central path with decreasing $\mu$ concurrently decreases the \emph{duality gap} $c^Tx - b^T y = x^T s = \mu m$.

As said above, path-following methods typically only maintain interior points that are sufficiently close to the central path. In this paper, we measure the closeness to the central path w.r.t.\ some $\mu >0$ by the relative deviation in the $1$-norm $||\sigma||_1$ where $\sigma \in \RR^m$ and $\sigma_a := \tfrac{x_a s_a}{\mu} - 1$. A sufficiently small deviation, e.g., $||\sigma||_1 \le \delta$ for some constant $\delta < 1$, guarantees that both $x_a$ and $s_a$ remain positive because $||\sigma||_1 \le \delta$ implies that $x_a s_a / \mu \ge 1 - \delta$ is well separated from $0$. However, this does not prohibit that either $x_a$ or $s_a$ becomes tiny while the other becomes huge. Such behavior is not only numerically unpleasant, it is also an obstacle for using fixed-precision integer arithmetic.

\subsection{Minors to the Rescue}\label{subsec:minors}

We deal with the issue mentioned above in a somewhat radical way: We only consider the arcs with sufficiently large values for $x_a$ and $s_a$. That is, whenever $x_a$ becomes too small, we delete the arc from the network and whenever $s_a$ becomes too small, we contract the arc $a$. We thereby obtain a minor of the original network. This is captured formally by the following definition.

\begin{definition}\label{def:minor}
	Let $G=(V,A)$ be a directed graph with $m$ arcs. For $\eps > 0$ and $x,s \in \RR^m_{\ge 0}$, we call the minor $H$ of $G$ obtained by deleting the arcs in $D := \{ a \in A: x_a < \tfrac{\eps}{m} \}$ and contracting the arcs in $C := \{ a \in A: s_a < \tfrac{\eps}{m} \}$ the $\eps$-$x$-$s$-minor of $G$.
\end{definition}

This approach might be surprising, but it has its foundation in classical min-cost flow theory~\cite[Section 10.6]{DBLP:books/daglib/0069809}.
We provide a brief reasoning for being self-contained. We first review a well-known useful fact from general LP duality.
\begin{lemma}\label{lem:cequivs}
	Let $\min\{c^Tx:Ax=b,x\ge 0\}$ be bounded and feasible and let $\max\{b^Ty:A^Ty+s=c,s\ge 0\}$ denote its dual linear program.
	\begin{itemize}
		\item A primal optimal solution $x^*$ is also an optimal solution of $\min\{\bar{s}^T x: Ax=b,x\ge 0\}$ for any dual feasible solution $\bar{y},\bar{s}$.
		\item A dual optimal solution $y^*,s^*$  is also an optimal solution of $\min\{\bar{x}^T s: A^Ty + s =c, s\ge 0\}$ for any primal feasible solution $\bar{x}$.
	\end{itemize}
\end{lemma}
\begin{proof}
	Let $x^*$ and $y^*,s^*$ be primal and dual feasible solutions, respectively.
	\begin{itemize}
		\item Let $\bar{y},\bar{s}$ be a dual feasible solution. Define $y' := y^*- \bar y$ and $s' := s^* = c - A^Ty^*$, then $A^T y' + s' = A^T y^* - A^T \bar{y} + s^* = c - A^T \bar{y} = \bar s$ and hence $y', s'$ are feasible for $\max\{b^Ty:A^Ty+s=\bar{s},s\ge 0\}$. Moreover, it holds that
		      ${x^*}^T s' = {x^*}^T s^* = 0$.
		      It follows that $x^*,y',s'$ are primal and dual optimal for $\min\{s^Tx':Ax'=b,x'\ge 0\}$ and $\max\{b^Ty':A^Ty'+s'=s,s'\ge 0\}$, respectively. This shows the claim.
		\item Let $\bar x$ be a primal feasible solution, i.e., $A\bar{x}=b$ holds. Suppose for contradiction that
		      \begin{align*}
		      	\bar{x}^T s^* & > \min\{\bar{x}^T s :A^Ty + s=c,s\ge 0\} = \min\{\bar{x}^T (c - A^Ty): A^Ty + s=c,s\ge 0\} \\
		      	              & = c^T \bar{x} - \max\{b^Ty: A^Ty + s=c,s\ge 0\} = c^T \bar{x} - b^T y^*
		      	= \bar{x}^T s^*.\tag*{\qedhere}
		      \end{align*}
	\end{itemize}
\end{proof}

The following lemma shows the rationale for deleting arcs with tiny flow $x$ and thus huge reduced costs $s$ in the vicinity of the central path (respectively, contracting arcs with tiny reduced costs and huge flow). \footnote{Note that this result can be extended to general integral linear programs, but we prefer to keep it simple for the sake of presentation.}
\begin{lemma}\label{lem:removecontract}
	Let $x,s \in \RR^m$ be primal/dual feasible solutions of $(G,b,c)$ with $b \in \ZZ^n$ and $c \in \ZZ^m$.
	\begin{itemize}
		\item If $s_a > x^Ts$ for some $a\in A$, then $x^*_a=0$ for every optimal solution $x^*$.
		\item If $x_a > x^Ts$ for some $a\in A$, then $s^*_a=0$ for every optimal solution $s^*$.
	\end{itemize}
\end{lemma}
\begin{proof}
	Because of the total unimodularity of the node-arc incidence matrix $A$ all basic solutions are integral. Thus, it suffices to show the claims for optimal integral basic solutions. All other optimal solutions are convex combinations of these.
	\begin{itemize}
		\item
		      Assume for contradiction that $x^*_a>0$ in an optimal integral solution, i.e., $x^*_a\ge 1$. Then ${x^*}^Ts \ge x^*_a s_a > x^Ts$, this shows that $x^*$ is not optimal for $\min\{s^Tx:Ax=b,x\ge 0\}$, contradicting Lemma~\ref{lem:cequivs}.
		\item Assume for contradiction that $s^*_a>0$ in an optimal integral solution, i.e., $s^*_a\ge 1$. Then
		      $x^T{s^*} \ge x_a s^*_a > x^Ts$, this shows that $s^*$ is not optimal for $\min\{x^Ts:A^Ty+s=c,s\ge 0\}$, contradicting Lemma~\ref{lem:cequivs}.\qedhere
	\end{itemize}
\end{proof}

Note that the central path condition $\sum_{a\in A} |\tfrac{x_as_a}{\mu}-1|\le \delta$ implies $x_as_a\in [(1-\delta)\mu, (1+\delta)\mu]$ for any arc $a\in A$. Now assume that for some arc $a\in A$, the value of $x_a$ becomes tiny, i.e., $x_a< \tfrac{\eps}{m}$.
Then
$
s_a\ge \tfrac{(1-\delta)\mu}{x_a}> \tfrac{(1-\delta)\mu m}{\eps} \ge (1+\delta)\mu m \ge x^Ts
$
follows for $\Eps := \tfrac{1-\delta}{1+\delta} < 1$. Thus Lemma~\ref{lem:removecontract} applies and we can conclude that $x^*_a=0$ in every optimum primal solution $x^*$. A completely symmetric approach shows that, if $s_a<\tfrac{\eps}{m}$ for an arc $a\in A$, then $s^*_a=0$ in every  optimal dual solution $s^*$.

The analyses of typical classical combinatorial algorithms that use arc deletion or arc contraction, respectively, are based on an argument that termination is guaranteed because at least one arc has to deleted or contracted after a certain amount of iterations and there are only $m$ arcs.
However, we have a different termination criterion that mainly depends on the duality gap. In~\cite{DBLP:conf/isaac/BeckerK14}, we showed that it suffices to compute a pair of primal/dual interior points with duality gap strictly less than $1$ to efficiently perform a crossover to an optimum integral basic feasible solution. In this paper, we extend this result and show that proximity to the optimum solution w.r.t.\ the minor suffices. To this end, we first formally define what proximity means in our context.

\begin{definition}\label{def:proxy}
	A pair of primal and dual feasible solutions $x,s$ of $(G,b,c)$ with $b \in \ZZ^n$ and $c \in \ZZ^m$ is called a proxy for the optimum of $(G,b,c)$ if $\sum_{a \in A_H} x_a s_a < (1-\eps)^2$ where $A_H$ is the arc-set of an $\eps$-$x$-$s$-minor of $G$ for some $\eps < 1$.
\end{definition}

\subsection{The Algorithm}\label{subsec:algorithm}
In Appendix~\ref{sec:initial}, we show how to construct, in linear time, initial interior points for an auxiliary instance with the same optimum that fulfill the central path condition $\|\sigma\|_1\le \delta$ for $\mu_0=3m\beta\gamma UC/\delta$ and thus can be used as input for Algorithm~\ref{alg:pathfollow}.
The idea is to decrease the current $\mu$ by roughly a factor of $1-\delta/\sqrt{m}$, more precisely, we set $\mu\leftarrow \lceil (1-\tau)\mu\rceil$, where $\tau=\delta/\sqrt{m}$, as long as we have not achieved a solution that is a proxy for the optimum of $(G,b,c)$ yet. Then, we restore closeness to the central path w.r.t.\ the new $\mu$ and the minor in a so-called \emph{centering step}.
Observe that the centering step and the absence of tiny $x_a$ and $s_a$ in the minor implies that their respective counterparts are bounded by $(1+\delta)\mu/\eps$. Hence, the magnitude of the initial $\mu$ yields an upper bound on the numbers that have to be considered in this algorithm.

Moreover, let $\beta := \gcd(b)$ and $\gamma := \gcd(c)$. Observe that the instance $(G,b/\beta,c/\gamma)$ still has integral demands and costs and it is combinatorially equivalent to the instance $(G,b,c)$. However, we do not scale the instance down, but we rather maintain this granularity to restrict the computations to integer arithmetic, i.e., we follow the central path by hopping from one point of the integer grid to an other one. We show that this is possible for sufficiently large $\beta$ and $\gamma$ and that the running time only scales logarithmically in $\beta$ and $\gamma$ so that up-scaling to meet this condition can be achieved at low cost. More precisely, suppose that a given input instance does not satisfy our assumption of having sufficiently large $\gcd$'s. Then, the approach is to scale down the instance by dividing by $\gcd(b)$ and $\gcd(c)$, respectively, and scaling them up by sufficiently large numbers that are given by the following theorem.

\begin{theorem}\label{thm:proxynumbersize}
	Let $(G,b,c)$ be a min-cost flow instance with $b \in \beta \ZZ^n$ and $c \in \gamma \ZZ^m$. There is a randomized algorithm to compute $x,s \in \ZZ^m$ s.t.\ $x/\beta,s/\gamma$ is a proxy for the optimum of $(G,b/\beta,c/\gamma)$ in expected $\tilde O( m^{3/2} )$ time using only integer arithmetic on numbers bounded by $2^8 m^{3} \gamma UC$,
	provided that $\beta \ge 2^8m^3$ and $\gamma \ge 2^{15}m^4\beta UC$.
\end{theorem}

\providecommand{\potshift}{\Delta}
\begin{figure}[ht]
	\begin{center}
		\begin{minipage}[t]{0.59\textwidth}
			\vspace{0pt}
			\begin{algorithm}[H] \DontPrintSemicolon \SetKwInOut{Input}{Input}\SetKwInOut{Output}{Output}
				\Input{$(G,b,c)$ with $\gcd(b) = \beta$ and $\gcd(c) = \gamma$, feasible interior $x,s \in \ZZ^m$, and $\mu>0$
					s.t.\ $\|\sigma\|_1:=\sum_{a\in A}|\tfrac{x_as_a}{\mu}-1|\le \delta$.}
				\Output{Vectors $x,s$ s.t.\ $x/\beta$ and $s/\gamma$ yield a proxy for the optimum of $(G,b/\beta,c/\gamma)$.}
				\Repeat{$\sum_{a\in A_H}x_as_a < (1-\eps)^2 \beta \gamma$}{
					$H$ = \texttt{minor}($G,x,s$)\;
					$\mu\leftarrow \lceil(1-\tau)\mu\rceil $, where $\tau = \delta/\sqrt{m}$\;
					$x,s$ = \texttt{centering\_step}($H,x,s,\mu$)\;
				}
				\Return{
					$x,s$
				}
				\caption{Integer Interior Point Algorithm\label{alg:pathfollow}}
			\end{algorithm}
		\end{minipage}%
		\hspace{0.00001\textwidth}
		\begin{minipage}[t]{0.38\textwidth}
			\vspace{0pt}
			\begin{algorithm}[H] \DontPrintSemicolon
				Let $\Eps:=\frac{1-\delta}{1+\delta}$.\;
				\For{$a\in A$}{
					\lIf{$x_a < \tfrac{\Eps\beta}{m}$}{
						remove $a$ from $G$
					}
					\lIf{$s_a < \tfrac{\Eps\gamma}{m}$}{
						contract $a$ in $G$
					}
				}
				\Return{$G$}
				\caption{\texttt{minor}($G,x,s$)}
				\label{minor}
			\end{algorithm}
		\end{minipage}
	\end{center}
\end{figure}

We sketch the algorithm that achieves the above theorem as pseudo-code in Algorithms~\ref{alg:pathfollow} and~\ref{minor}. 
The randomization only happens in the subroutine for the centering step, where fundamental cycles w.r.t.\ a low-stretch spanning tree are sampled. We discuss the centering step in detail in Section~\ref{sec:centering step}. For now it is only important that it maintains closeness to the central path w.r.t.\ the arcs in the minor. For the sake of presentation, we did not maintain the $y$-variables that can be used to restore the reduced costs on the arcs that have been deleted. But this can be easily achieved. Moreover, one can extend the cycles mentioned above from the minor through the contracted nodes to cycles in the original graph such that one can also keep track of the updates to maintain a feasible solution in the original graph. We thereby can obtain a proxy for the optimum of $(G,b,c)$. However, the variables of the arcs in the minor describe an interior point of a slightly perturbed instance.

\subsection{Perturbation}\label{subsec:perturbation}

Suppose, for example,  that we have an $\eps$-$x$-$s$-minor $H$ that differs from $G$ just by the deletion of a single arc $\hat{a}$. Hence, the set of nodes of $H$ coincides with the one of $G$.
However, projecting $x$ to $x_H$, $s$ to $s_H$, and $c$ to $c_H$ by removing the entry corresponding to $\hat{a}$ does not yield a pair of primal/dual feasible solutions $x_H,s_H$ for $(H,b,c_H)$, but it is rather a feasible solution to $(H, \tilde b, c_H)$ where $\tilde b = b - A e_{\hat a} e_{\hat a}^T x$ and we denote with $e_{\hat a}$ the vector of unity corresponding to the index $\hat a$. Note that the two demands vectors $b$ and $\tilde b$ are almost identical because $x_{\hat a}$ is tiny. We first define formally what we mean by almost identical instances and then draw the connection to minors.

\begin{definition}
	An instance $(G,\tilde b,\tilde c)$ is called \emph{$\eps$-perturbed} instance of $(G,b,c)$, if $\|b-\tilde b\|_1\le 2\eps$ and $\|c-\tilde c\|_1\le \eps$.
\end{definition}

We introduce the following notation to transform instances and solutions:
\begin{itemize}
	\item Denote with $\Ib\in\RR^{n_H\times n}$ the linear map that transforms $b$ into $b_H$, where for every node $v\in V_H$, $b_H(v)=\sum_{u\in C_v}b(u)$ for $C_v\subseteq V$ being all nodes in $G$ that get contracted to $u$.
	\item As above for any $a\in A$, we denote with $e_a\in\{0,1\}^m$ the vector of unity that has a $1$ in its $a$'th component and $0$ elsewhere, where the arcs are ordered as the columns of the node-arc-incidence matrix $A$. Similarly, with $e_v\in\{0,1\}^n$ we mean the vector of unity that has a $1$ in its $v$'th component, $0$ elsewhere and the considered ordering is as in the rows of $A$.
	\item Denote as $\MH\in \{0,1\}^{m_H\times m}$ the map with the $a$'th row being $e_a$ for all $a\in A_H$.
	\item Denote as $\NH\in \{0,1\}^{n_H\times n}$ the map with the $v$'th row being $e_v$ for all $v\in V_H$.
	\item For some $S \subseteq A$, let $I_S := \sum_{a \in S}e_a e_a^T$.
\end{itemize}

\begin{lemma}\label{lem:eps perturbation}
	Let $x,s$ be a pair of primal/dual feasible solutions of $(G,b,c)$ let $C:=\{a\in A_G: s_a < \tfrac{\eps}{m}\}$ and $D:=\{a\in A_G: x_a < \tfrac{\eps}{m}\}$ for some $0<\eps<1$. Then $(G,\tilde b,\tilde c)$ with $\tilde b = b - A \ID x$ and $\tilde c = c - \IC s$ is an $\epsilon$-perturbed instance of $(G,b,c)$.
\end{lemma}
\begin{proof}
	The proof is a simple calculation:
	\begin{align*}
		\|b-\tilde b\|_1 = \|A\ID x\|_1 \le 2 \|\ID x\|_{1} \le 2\eps \quad\text{ and }\quad
		\|c-\tilde c\|_1 = \|\IC c\|_1 \le \eps. \tag*{\qedhere}
	\end{align*}
\end{proof}

The following Theorem shows the equivalence of the original and the perturbed problem. For a min-cost flow instance $(G,b,c)$ and a given spanning tree $T$ of $G$, we call $x\in \RR^m$ a \emph{primal tree solution}, if $Ax=b$ and $x_a=0$ for all $a\in A\setminus T$ and we call $y,s\in\RR^{n+m}$ a \emph{dual tree solution} if $A^Ty + s =c$ and $s_a=0$ for all $a\in T$.

\begin{theorem}\label{thm:trees}
	Let $(G,\tilde b,\tilde c)$ be an $\eps$-perturbed instance of $(G,b,c)$  for some $0<\eps<1$ and let $T$ be a spanning tree of $G$.
	\begin{enumerate}[a)]
		\item If the unique tree solution w.r.t.\ $T$ in $(G,\tilde b,\tilde c)$ is primal feasible, then the unique primal tree solution w.r.t.\ $T$ in $(G,b,c)$ is feasible.
		\item If the unique tree solution w.r.t.\ $T$ in $(G,\tilde b,\tilde c)$ is dual feasible, then the unique dual \footnote{Here \emph{unique dual} refers only to the reduced costs because $A^T(y + t\ones) = A^T y$ for all $t \in \RR$.} tree solution w.r.t.\ $T$ in $(G,b,c)$ is feasible.
		\item If the spanning tree $T$ is optimal for $(G,\tilde b,\tilde c)$, then it is optimal for $(G,b,c)$.
	\end{enumerate}
\end{theorem}
\begin{proof}
	\begin{enumerate}[a)]
		\item\label{item:treethmprimal} Let $x_T$ and $\tilde x_T$ denote the unique primal tree solutions corresponding to $T$ in $(G,b,c)$ and $(G,\tilde b,\tilde c)$, respectively, i.e., $A_Tx_T=b$ and $A_T\tilde x_T=\tilde b$ holds. Similarly let $d$ be the unique primal tree solution corresponding to $T$ for the right hand side $b-\tilde b$, i.e., $A_Td=b-\tilde b$. Note that in order to show feasibility of $x_T$, it suffices to show that $x_T\ge 0$.
		It follows that
		\begin{align*}
			A_T d = b - \tilde b = A_T(x_T-\tilde x_T)
		\end{align*}
		and thus by the uniqueness of $d$, we conclude $d=x_T-\tilde x_T$. It follows that $\|x_T-\tilde x_T\|_{\infty} = \|d\|_{\infty}\le \|b-\tilde b\|_1/2\le \eps$, using that $(G,\tilde b,\tilde c)$ is an $\eps$-perturbed instance. However, since $b\in\ZZ^n$ also $x_T\in\ZZ^m$ and since $\eps<1$ it follows that $\tilde x_T\ge 0$ implies $x_T\ge 0$.
		\item\label{item:treethmdual} Let $y_T,s_T$ and $\tilde y_T,\tilde s_T$ denote the unique dual tree solutions corresponding to $T$ in $(G,b,c)$ and $(G,\tilde b,\tilde c)$, respectively, i.e., $A^Ty_T + s_T=c$, $s_T(a)=0$ for $a\in T$ and $A^T\tilde y_T + \tilde s_T=\tilde c$, $\tilde s_T(a)=0$ for $a\in T$ holds.
		Similarly let $p,d$ be the unique dual tree solution corresponding to $T$ for the right hand side $c-\tilde c$, i.e., $A^Tp+d=c-\tilde c$.
		Note that in order to show feasibility of $y_T,s_T$, it suffices to show that $s_T\ge 0$. It follows that
		\begin{align*}
			A^Tp+d = c - \tilde c = A^T(y_T- \tilde y_T) + s_T - \tilde s_T
		\end{align*}
		and thus by the uniqueness of $d$, we conclude $d=s_T-\tilde s_T$. Furthermore for $a\in T$, we have $d(a)=0$ and for $a\in A\setminus T$, it holds that $|d_a|=|c_a-\tilde c_a-(p_w-p_v)|=|c_a-\tilde c_a- \sum_{b\in P(v,w)}c_b-\tilde c_b|\le \|c-\tilde c\|_1$. It follows that $\|s_T-\tilde s_T\|_{\infty} = \|d\|_{\infty} \le \|c-\tilde c\|_1\le \eps$, using that $(G,\tilde b,\tilde c)$ is an $\eps$-perturbed instance. However, since $c\in\ZZ^m$ also $s_T\in\ZZ^m$ and since $\eps<1$ it follows that $\tilde s_T\ge 0$ implies $s_T\ge 0$.
		\item Assume the tree $T$ is optimal for $(G,\tilde b,\tilde c)$, i.e., the tree is primal and dual feasible (complementary slackness). Together with part \ref{item:treethmprimal} and \ref{item:treethmdual}, we conclude that $T$ is also optimal for $(G,b,c)$. \qedhere
	\end{enumerate}
\end{proof}
We note that in the non-degenerate case, the reverse statements of \ref{item:treethmprimal} and \ref{item:treethmdual} hold as well. As our algorithm will output a tree solution to the perturbed problem, the above
Theorem is sufficient to show that the output tree is also an optimal tree for the original problem.

\subsection{Crossover}\label{subsec:crossover}
The crossover algorithm from~\cite[Section 2]{DBLP:conf/isaac/BeckerK14} takes a min-cost flow instance $(G,\tilde b,\tilde c)$ and primal/dual feasible points $\tilde x,(\tilde y,\tilde s)$. Given that $\tilde c$ is integral, it outputs an integral dual tree solution $y^*, s^*$ with $b^T y^* \ge b^T\tilde y$.
In a nutshell, the crossover algorithm works as follows. It iteratively constructs nested cuts $\{s\}=S^1\subset S^2 \subset \ldots \subset S^n=V$, starting with an arbitrary node $s$. During iteration $i$, the algorithm modifies the potentials $y$ along the cut $S^i$ ensuring that $b^Ty$ does not decrease. This is done by incrementing the potentials of all nodes in $S^i$, if $b(S^i)\ge 0$, and by decrementing them otherwise. The potentials are increased (decreased) until the reduced cost of at least one arc $a$ on the cut $S,V\setminus S$ becomes 0. Then, the node adjacent to $S^i$ through $a$ is being added to $S^i$. We note that the amount by which the potentials inside of $S^i$ are changed relatively to the nodes outside of $S^i$, is exactly equal to the reduced cost of the arc $a$.

\begin{theorem}
	Let $(G,b,c)$ be a min-cost flow instance with $b \in \beta \ZZ^n$ and $c \in \gamma \ZZ^m$. Given $x,s \in \ZZ^m$ s.t.\ $x/\beta,s/\gamma$ is a proxy for the optimum of $(G,b/\beta,c/\gamma)$, a pair of primal-dual optimum solutions $x^*,s^*$ of $(G,b,c)$ can be computed in $O( m^{3/2} \log \tfrac{n^2}{m} \log U )$ time.
\end{theorem}

\begin{proof}
	Let $\bar x=x/\beta$, $\bar s=s/\gamma$ as well as $\bar b=b/\beta$ and $\bar c=c/\gamma$.
	Consider the projection $\tilde x$ of $\bar x$ to the subspace with $\tilde x_a = 0$ for all $a \in D$ and the projection $\tilde s$ of $\bar s$ to the subspace with $\tilde s_a = 0$ for all $a \in C$.
	Note that these projections are primal/dual feasible for the $\eps$-perturbation $(G,\tilde b, \tilde c)$ of $(G,b,c)$, where $\tilde b = \bar b-AI_Dx$ and $\tilde c = \bar c - I_C s$ and their duality gap satisfies $\tilde x^T \tilde s = \sum_{a \in A_H} \bar x_a \bar s_a < (1-\eps)^2$. Given feasible $\bar s$, we can compute corresponding feasible $\bar y$ in linear time by propagation from the root to the leaves of an arbitrary spanning tree. 
	The same holds for $\tilde y$ and $\tilde s$. If we perform the crossover procedure with $\tilde y, \tilde s$ in $(G,\tilde b, \tilde c)$, we obtain a tree solution w.r.t.\ some spanning tree $T$. By Theorem~\ref{thm:trees}, this spanning tree also yields a dual feasible solution for $(G,b,c)$. Moreover, the two admissible networks are combinatorially the same. If the admissible network yields a primal feasible solution w.r.t.\ $b$, we obtain a pair of primal/dual feasible solutions satisfying complementary slackness and hence the sought optimum solutions $x^*,s^*$. Suppose the contrary for a contradiction. Then, there is a cut $S \subset V$ s.t.\ $b(S) > 0$ and there are no ingoing arcs with vanishing reduced costs, i.e., all ingoing arcs have reduced costs of at least $1$. Thus, $y_v$ could be safely increased by $1$ for all $v \in S$. Moreover, $\tilde y_v$ could be safely increased by $1-\eps$. Thus, the dual objective $\tilde b^T \tilde y$ would increase by at least $(1-\eps)^2$ because $\ones_S^T \tilde b = \ones_S^T b - \ones_S A \ID x \ge 1 - \eps$. But this contradicts $\tilde x^T \tilde s < (1-\eps)^2$. The crossover can be computed in $O( m + n \log n)$ and the transshipment problem in the admissible network that needs to be solved in order to obtain the corresponding primal solution can be solved by a max-flow computation in $O( m^{3/2} \log \tfrac{n^2}{m} \log U )$ using the algorithm of Goldberg and Rao~\cite{DBLP:conf/focs/GoldbergR97}.
\end{proof}
Note that the additional max-flow computation can be avoided, by using the isolation lemma~\cite{DBLP:journals/combinatorica/MulmuleyVV87} in order to make the optimal solution unique~\cite[Lemma 3.12]{DBLP:conf/stoc/DaitchS08}. However, the perturbation of the cost vector yields non-integral values and, hence, an additional scaling would be necessary in order to make the input integral. Furthermore, this strategy would introduce a new source of randomization.

\section{Centering Step}\label{sec:centering step}
In this section, we describe how to implement the centering step by using a variant of the electrical flow solver from~\cite{DBLP:conf/stoc/KelnerOSZ13}. We would like to stress the fact that we do not rely on the numerical stability analysis of~\cite{DBLP:conf/stoc/KelnerOSZ13} that would involve an additional scaling of the current sources, instead the current sources that are used here can be directly defined as integer values dependent on the current iterates $x,s$.
The centering step takes as input the minor $H$ of $G$, the variables $x,s$ and the parameter $\mu'=\lceil (1-\tau)\mu \rceil$, with $x$ and $s$ fulfilling $\|\sigma\|_1\le \delta$, where $\sigma_a = x_as_a/\mu-1$ for all $a\in A_H$. For simplicity, we will denote in this subsection the vectors $x$ and $s$ as vectors over $A_H$, instead of over $A$.
The goal of \texttt{centering\_step} is to update $x$ and $s$ to $x+\xupd, s+\supd$ such that $\|\sigma'\|_1\le \delta$ holds, where $\sigma'_a=\frac{x'_as'_a}{\mu'}-1$.
Note that $x,s$ uniquely defines feasible potentials $y$ whose update we denote with $\yupd$. Since we want the new iterates to be interior solutions again, we require
\begin{align}
	\label{form:eqfeasible}
	A_H \xupd =0, \; A_H^T\yupd+\supd=0 \text{ and } x+\xupd>0 \text{ as well as }s+\supd>0.
\end{align}
The idea is to compute $\xupd,\supd$ such that $(x+\xupd)(s+\supd)\approx \mu'$, more precisely we will guarantee that $\|\devfromone\|_1\le \delta/4$, where $\devfromone:=\frac{1}{\newmu} \left[S\xupd + X\supd +Xs \right] - \ones$, i.e., we omit the quadratic term $\xupd\supd$ and allow an approximate solution only. Here $X=\diag(x)$ and $S=\diag(s)$, respectively.
It is easy to see that computing such $\xupd, \supd$ satisfying~\eqref{form:eqfeasible} amounts to approximately solving a linear equation system. This approximate solution to the equation system in turn can be found by computing an approximate electrical flow in an electrical network with resistances $r_a=\lceil s_a/x_a\rceil$ and current sources $A_H\f^0$, where $\f^0_a=x_a- \lceil \mu'/s_a \rfloor$ for all $a\in A_H$. We will denote $R:=\diag(r)$. Note that, due to the assumption on the size of $\beta$ and $\gamma$, it holds that $s_a/x_a\ge 1$, see the proof of Lemma~\ref{lem:gapclose}.

We need some additional notation concerning electrical flows.\footnote{For a more detailed depiction of electrical flows see for example~\cite{bollobas1998modern,DBLP:conf/stoc/KelnerOSZ13}.} Let $\f$ be a given flow and let $T$ be a spanning tree of $H$. For any $a=(v,w)\in A_H\setminus T$, we define $C_a:=\{a\}\cup P(w,v)$, where $P(w,v)$ is the unique path in $T$ between $w$ and $v$, and $r(C_a) := \sum_{a'\in C_a} r_{a'}$ and $\tcn(T):=\sum_{A_H\setminus T} \tfrac{r(C_a)}{r_a}$ as the \emph{tree condition number of $T$}. The \emph{tree induced voltages} are defined as $\p_v:=\sum_{a\in P(v_0,v)} \f_a r_a$, where $v_0$ is an arbitrary root of the spanning tree $T$.

As previously mentioned, our method can be seen as an adaption of the electrical flow solver of Kelner et al.\ that was introduced in~\cite{DBLP:conf/stoc/KelnerOSZ13}. Kelner et al.\ showed how to compute an approximate electrical flow $\f$ and corresponding node voltages $\p$ in nearly linear time. Up to the choice of the termination criterion and the choice of the initial current our algorithm is identical, see the pseudo-code implementation for further details. In the algorithm, we need to compute a spanning tree and in order to gurantee fast convergence this spanning tree needs to be of low stretch. The current record bounds concerning low stretch spanning trees is achieved by the algorithm of Abraham and Neiman~\cite{DBLP:conf/stoc/AbrahamN12}.
They show how to compute a tree of stretch $O(m\log n\log\log n)$ in $O(m \log n \log \log n)$ time.

\begin{algorithm}[H] \DontPrintSemicolon
	\BlankLine
	Define $r_a:=\lceil \frac{s_a}{x_a}\rceil$,
	$\f^0_a:= x_a- \lceil \tfrac{(1-\tau)\mu}{s_a} \rfloor$ for $a\in A_H$.\;
	Let $T:=$spanning tree and
	$p_a:= \tfrac{r(C_a)}{\tcn(T)r_a}\forall a\in A_H\setminus T$\;
	$\f:=\f^0$, $s' = s$ and $x' = x$ \;
	\While{$ \sum_{a\in A_H} |x'_a s'_a - \mu| \ge \delta\mu$
		}{
		Sample $a\in A_H\setminus T$ according to $p$\;
		$\alpha:=\lceil \tfrac{ -\sum_{a'\in C_a} r_{a'} \f_{a'}}{r(C_a)} \rfloor $,
		$x' \leftarrow x' + \alpha c_a$, $\f \leftarrow \f+ \alpha c_a$\;
		occasionally: compute tree induced voltages $\p$ for $\f$, $s' := s-A_H^T\p$\;
	}
	\Return{$x',s'$}
	\caption{\texttt{centering\_step}($H,x,s,\mu$)}
	\label{alg:centering_step}
\end{algorithm}
\providecommand{\taupr}{\tau'}
Note that for our purposes here, we will assume the parameter $\delta$ to be set to $1/8$ and we will denote with $x,s,\mu$ the values of the variables at the beginning of the \texttt{centering\_step} and the update values we define by $\xupd := \f - \f^0$ and $\supd := -A_H^T\p$. We remark that, for our choice of $\delta=1/8$, evaluating the condition of the while-loop can be done easily with integer arithmetic (by comparing the two integral values resulting from multiplying the inequality by 8).

Moreover, define $\taupr$ by $(1-\taupr)\mu=\lceil (1-\tau)\mu\rceil=\newmu$ and recall that we set $\Eps:=\tfrac{1-\delta}{1+\delta}$.
Note that from the minorization and closeness to the central path, we obtain lower and upper bounds on $x$ and $s$ for arcs in the minor, respectively:
\begin{align}\label{form:lowerupperbounds}
	x_a \in \left[\frac{\Eps\beta}{m}, \frac{(1+\delta)\mu m}{\Eps \gamma}\right] \quad \text{ and }\quad
	s_a \in \left[\frac{\Eps\gamma}{m},\frac{(1+\delta)\mu m}{\Eps \beta}\right] \quad \text{ for all }a\in A_H.
\end{align}
We first of all show that the updates are feasible moves. The constraints $A_Hx'=b$ follow since $\xupd$ is a circulation. The update for potentials can be computed in such a way that the dual constraints, except non-negativity of $s$, are fulfilled as well. For non-negativity of $x$ and $s$, consider the following lemma - in fact the closeness to the central path already implies non-negativity for $x$ and $s$, respectively.
\begin{lemma}
	If $\|\sigma'\|_1 < \delta$, then $x_a + \xupd_a > 0$ and $s_a + \supd_a > 0$ for all $a\in A_H$.
\end{lemma}
\begin{proof}
	If $\|\sigma'\|_1 < \delta$, we must have $|(x_a + \xupd_a)(s_a + \supd_a)/{\mu'} - 1| < \delta$ for all $a\in A_H$. Thus $(x_a + \xupd_a) (s_a + \supd_a) > 0$ for all $a\in A_H$. Assume $x_a + \xupd_a \le 0$ and $s_a + \supd_a \le 0$ for some $a\in A_H$. Then
	\begin{align*}
		0 \ge s_a (x_a + \xupd_a) + x_a (s_a + \supd_a) = \mu' + s_a x_a > 0, \quad\text{a contradiction.} \tag*{\qedhere}
	\end{align*}
\end{proof}
In the following lemma we show that, assuming sufficiently large $\beta$ and $\gamma$, a small $\gap(\f,\p):=\f^TR\f -2\p^TA_H\f^0 + \p^TA_HR^{-1}A_H^T\p$ of the electrical flow and voltages yields small one-norm of $\devfromone$.
\begin{lemma}\label{lem:gapclose}
	Assume $\beta\ge \tfrac{2^4m^2}{\delta}$ and $\gamma\ge \tfrac{2^{17}m^5 \beta UC}{\delta}$. Let $\delta\le 1/8$ and $\f\in \RR^{m_H}$ and $\p\in\RR^{n_H}$ such that $\gap(\f,\p)<\tfrac{2^{-8}\delta^2\mu}{m_H}$. Then it holds that $\|\devfromone\|_1\le \tfrac{\delta}{4}$, where $\devfromone:=\frac{1}{\newmu} \left[S\xupd + X\supd +Xs \right] - \ones$.
\end{lemma}
\begin{proof}
	In order to simplify notation, we will use $x$ and $s$ for $M_H x$ and $M_H s$, respectively in this proof, since all consideration are concerned with the minor $H$.
	Let us denote with $D_a:=\lceil \tfrac{\newmu}{s_a}\rfloor - \tfrac{\newmu}{s_a}$ and $E_a := \lceil \tfrac{s_a}{x_a}\rceil - \tfrac{s_a}{x_a}$ the errors introduced by the rounding in the algorithm.
	Then,
	\begin{align}\label{form:epsestimate}
		\begin{split}
		\|\devfromone\|_1
		  & =
		\sum_{a\in A_H}
		\frac{x_a}{\newmu} \left|\frac{s_a}{x_a}\xupd_a +\supd_a + s_a -\frac{\newmu}{x_a} \right|\\
		&=
		\sum_{a\in A_H}
		\frac{x_a}{\newmu} \left|\frac{s_a}{x_a}(\f_a-\f^0_a) - (\p_w-\p_v) + s_a -\frac{\newmu}{x_a} \right|\\
		  & =
		\sum_{a\in A_H}
		\frac{x_a}{\newmu}
		\left| r_a\f_a - (\p_w-\p_v)
		+ \frac{s_a}{x_a}D_a - E_a\f_a
		\right|\\
		&\le
		\frac{\|X(R\f-A^T\p)\|_1}{\newmu}
		+
		\frac{\|s\|_1}{2\newmu}
		+
		\sum_{a\in A_H}\frac{x_a|\f_a|}{\newmu},
		\end{split}
	\end{align}
	since $|E_a|<1$ and $|D_a|\le 1/2$. Moreover, note that using the upper and lower bounds on $x$ and $s$ in~\eqref{form:lowerupperbounds} and the fact that $\tau'\le \tau\le \delta$ yields the following multiplicative errors introduced by the rounding:
	\begin{align}\label{form:multroundingerror}
		r_a=\left\lceil \frac{s_a}{x_a}\right\rceil \in \left[\frac{s_a}{x_a}, 2\frac{s_a}{x_a}\right]\quad\text{and}\quad
		\left\lceil \frac{\mu'}{s_a} \right\rfloor\in\left[\frac{\mu'}{2s_a},\frac{2\mu'}{s_a}\right],
	\end{align}
	where we need that $s_a/x_a\ge 1$, which is guaranteed by $\gamma\ge 2^6 m^3 \beta UC$ and $\beta\ge 4m$.
	Let us now first show how to bound the last summand in \eqref{form:epsestimate}: Note that
	\begin{align*}
		\sum_{a\in A_H}\frac{x_a|\f_a|}{\newmu}
		&\le \frac{1}{\newmu}\sum_{a\in A_H}x_a\sqrt{\frac{x_a}{s_a}}\cdot \sqrt{\frac{s_a}{x_a}}|\f_a|
		\le \frac{1}{\newmu}\sqrt{\sum_{a\in A_H}\frac{x_a^3}{s_a}}
		\sqrt{\sum_{a\in A_H}\frac{s_a}{x_a}|\f_a^2|}\\
		&\le \frac{1}{\newmu}\sqrt{\sum_{a\in A_H}\frac{x_a^3}{s_a}}
		\|\f\|_R.
	\end{align*}
	Since the energy of the initial flow $\f^0$ is at least the energy of $\f$, we obtain
	\begin{align}\label{form:phinormestimate}
		\begin{split}
		\|\f\|_R^2
		  & \le
		\|\f^0\|_R^2
		\le
		\sum_{a\in A_H} \frac{2s_a}{x_a}\left[x_a^2-2 x_a\left\lceil \frac{\newmu}{s_a}\right\rfloor + \left\lceil \frac{\newmu}{s_a}\right\rfloor^2\right]
		\le
		\sum_{a\in A_H}2x_as_a-2\newmu+\frac{8\newmu^2}{x_as_a}\\
		  & \le
		m_H\left[2(1+\delta)\mu-2(1-\taupr)\mu+\frac{8\mu(1-\taupr)^2}{(1-\delta)}\right]
		\le
		\frac{1}{2}\mu+\frac{8\mu(1-\taupr)^2}{(1-\delta)}
		\le 10\mu m_H,
		\end{split}
	\end{align}
	where we used $x_as_a/\mu\ge 1-\delta$, $\delta\le \tfrac{1}{8}$ and $\taupr\le\tau=\delta/\sqrt{m}\le \delta$.
	The upper bound for $x$ and the lower bound for $s$ in~\eqref{form:epsestimate} yield $\sqrt{\sum_{a\in A_H}x_a^3/s_a}/\mu'\le \sqrt{(1+\delta)\mu m^5}/(\eps^3\gamma^2)$, again using $\taupr\le \delta$.
	Putting the two bounds together, we obtain $\sqrt{10(1+\delta)}\mu m^3/(\eps^3\gamma^2)$ as a bound on the third summand in~\eqref{form:epsestimate}.
	Now, for the first summand in~\eqref{form:epsestimate}, note that $\gap(\f,\p) = \sum_{a\in A_H\setminus T}[r_a\f_a - (\p_w -\p_v)]^2/r_a$, as shown in~\cite[Lemma 4.4]{DBLP:conf/stoc/KelnerOSZ13}.
	Hence,
	\begin{align*}
		\frac{\|X(R\f-A^T\p)\|_1}{\mu'}
		=
		\sum_{a\in A_H\setminus T} \frac{|r_a\f_a-(\p_w-\p_v)|}{\sqrt{r_a}}\frac{\sqrt{r_a}x_a}{\mu'}
		\le
		\frac{\sqrt{\gap\|X^2r\|_1}}{\mu'}
		\le \frac{2\sqrt{\gap m_H}}{\sqrt{\mu}}
	\end{align*}
	using the Cauchy-Schwarz-Inequality, $x_as_a/\mu\le 1+\delta$ and again $\taupr\le \delta$.
	Finally, for the second summand in~\eqref{form:epsestimate}, note that $\tfrac{\|s\|_1}{2\mu'}\le \tfrac{(1+\delta)\mu m^2}{2\mu' \Eps \beta}\le \tfrac{m^2}{2\Eps^2\beta}$.
	Hence, plugging all bounds together into~\eqref{form:epsestimate} yields
	\begin{align*}
		\|\devfromone\|_1
		\le
		\frac{2\sqrt{\gap m}}{\sqrt{\mu}}
		+ \frac{m^2}{2\Eps^2\beta}
		+ \frac{\sqrt{10(1+\delta)}\mu m^3}{\Eps^3\gamma^2}
		\le
		\frac{\delta}{8} + \frac{\delta}{16} + \frac{\delta}{16}
		\le
		\frac{\delta}{4},
	\end{align*}
	with $\beta\ge\tfrac{2^4m^2}{\delta}$, $\gamma\ge \tfrac{2^{12} m^4 \beta UC}{\delta}$, $\gap\le \tfrac{2^{-8}\delta^2\mu}{m_H}$ and the upper bound on $\mu\le\mu_0\le 24m\beta\gamma UC$.
\end{proof}
We next show that a small one-norm of $\devfromone$ actually implies the closeness to the central path of the new iterates $x'$ and $s'$, again assuming sufficiently large $\beta$ and $\gamma$.

\begin{lemma}\label{lem:invarsigma}
	Let $\delta\le 1/8$. If $\|\devfromone\|_1\le \frac{\delta}{4}$, then $\|\sigma'\|_1\le \delta$ where $\sigma'_a := \tfrac{x'_as_a'}{\newmu}-1$ for $a\in A_H$.
\end{lemma}
\begin{proof}
	By definition of $\devfromone$, it holds that $(1+\devfromone_a)\newmu = \xupd_a s_a + x_a \supd_a + x_as_a$, which yields
	\begin{align}\label{eq:xupdsupd}
		\xupd_a\supd_a
		= \tfrac{1}{2x_as_a}\left[ ((1+\devfromone_a)\newmu - x_as_a)^2 - (\xupd_as_a)^2 -(x_a\supd_a)^2\right] \quad\text{for }a\in A_H.
	\end{align}
	For the one-norm of $\sigma'$, we get the following estimate
	\begin{align*}
		\|\sigma'\|_1
		= \sum_{a\in A_H} \left| \frac{x'_as'_a}{\newmu} - 1 \right|
		= \sum_{a\in A_H} \left| \frac{\xupd_a\supd_a}{\newmu} + \devfromone_a \right|
		\le
		\sum_{a\in A_H} \left| \frac{\xupd_a\supd_a}{\newmu} \right| + \|\devfromone \|_1.
	\end{align*}
	Using~\eqref{eq:xupdsupd} and the triangle inequality, we obtain for the first summand that
	\begin{align*}
		\sum_{a\in A_H} \left| \frac{\xupd_a\supd_a}{\newmu} \right|
		  & \le
		\sum_{a\in A_H} \frac{[(1+\devfromone_a)\newmu-x_as_a]^2}{2x_as_a\newmu} + \frac{(\xupd_as_a)^2 + (x_a\supd_a)^2}{2x_as_a\newmu}\\
		&=
		\sum_{a\in A_H} \frac{[(1+\devfromone_a)\newmu-x_as_a]^2}{x_as_a\newmu},
	\end{align*}
	because $\xupd$ and $\supd$ are orthogonal and thus the sum over all arcs in~\eqref{eq:xupdsupd} yields the last equality.
	Now let $\taupr$ be such that $\mu'=(1-\taupr)\mu$. Factoring out $\mu$, plugging in the definition of $\newmu$ and $\sigma_a$ and using the lower bound on $\tfrac{x_as_a}{\mu}$ yields
	\begin{align*}
		\sum_{a\in A_H} \left| \frac{\xupd_a\supd_a}{\newmu} \right|
		  & \le
		\sum_{a\in A_H} \frac{\mu^2\big(\frac{(1+\devfromone_a)\newmu}{\mu}-\frac{x_as_a}{\mu}\big)^2}{x_as_a\newmu}
		\le
		\sum_{a\in A_H} \frac{[(1-\taupr)(1+\devfromone_a)- 1 - \sigma_a]^2}{(1-\taupr)(1-\delta)}.
	\end{align*}
	Expanding the square, re-grouping the terms and using $\ones^T\sigma\le\|\sigma\|_1\le \delta$, $-\ones^T\devfromone\le \|\devfromone\|_1$, the equivalence of norms and the Cauchy-Schwarz-Inequality for $-\devfromone^T\sigma\le\|\devfromone\|_2\|\sigma\|_2$, yields
	\begin{align*}
		\sum_{a\in A_H} \left| \frac{\xupd_a\supd_a}{\newmu} \right|
		  & \le
		\frac{(1-\taupr)^2 \|\devfromone\|_1^2 + [2(1-\taupr)\taupr +2\delta(1-\taupr)]\|\devfromone\|_1 + \delta^2 + 2\delta \taupr + \taupr^2m}{(1-\taupr)(1-\delta)}.
	\end{align*}
	Using $\taupr\le \tau=\tfrac{\delta}{\sqrt{m}}\le\delta$ and $1-\taupr\le 1$ yields
	\begin{align*}
		\|\sigma'\|_1
		\le
		\frac{\|\devfromone\|_1^2 + 4\delta\|\devfromone\|_1 + 4\delta^2 + (1-\taupr)(1-\delta)\|\devfromone\|_1}{(1-\taupr)(1-\delta)}
		\le
		\frac{\|\devfromone\|_1^2 + (4\delta+1)\|\devfromone\|_1 + 4\delta^2}{(1-\delta)^2}
	\end{align*}
	and $\|\devfromone\|_1\le \tfrac{\delta}{4}$ gives the result for any $\delta\le \tfrac{1}{8}$.
\end{proof}
It remains to argue the convergence behavior of \texttt{centering\_step}. The following lemma can be seen as the equivalence of \cite[Lemma 4.5]{DBLP:conf/stoc/KelnerOSZ13}. We obtain an additional factor of $3/4$ due to the rounding to the nearest integer of $\alpha$. We can show that assuming large enough $\beta$ and $\gamma$, the rounding does not impede the convergence speed significantly.
\providecommand{\Deltac}{\Lambda}
\begin{lemma}
	The expected decrease in primal energy $\|\f\|_R^2$ for every cycle-update in the inner while-loop of Algorithm~\ref{alg:pathfollow} is at least $\tfrac{3\gap}{4\tcn(T)}$, provided that $\beta\ge 2^5m^{3}/\delta$, $\gamma\ge 2^6m^3\beta UC$ and $\gap\ge 2^{-8}\delta^2\mu/m_H$.
\end{lemma}
\begin{proof}
	Let $\Deltac_a := \sum_{a'\in C_a} r_{a'} \f_{a'} = r_a \f_a - (\p_w-\p_v)$.
	In every iteration of the inner while-loop an arc $a$ is sampled and $\f$ is updated to $\f' = \f - \alpha_a \theta_a$, denoting $\alpha_a := \lceil \tfrac{ \Deltac_a }{r(C_a)} \rfloor$ and $\theta_a$ the indicator vector corresponding to $C_a$.
	The change in energy due to this update is $ \|\f - \alpha_a \theta_a \|_R^2 - \|\f\|_R^2 = - 2\alpha_a \Deltac_a + \alpha^2 r(C_a)$. Thus, using the definition of the probabilities $p_a=\tfrac{r(C_a)}{\tcn(T)r_a}$ for $a\in A_H\setminus T$, the expected decrease in energy is
	\begin{align*}
		\Exp \left[ \|\f'\|_R^2 - \|\f\|_R^2 \right]
		=
		\sum_{A_H\setminus T}
		\frac{r(C_a)^2 \alpha_a^2}{\tcn(T)r_a}
		-
		\frac{2r(C_a)\alpha_a \Deltac_a}{\tcn(T)r_a}
		=
		\sum_{A_H\setminus T}
		\frac{\big[\Deltac_a - \alpha_a r(C_a)\big]^2 - {\Deltac_a}^2}{\tcn(T)r_a}.
	\end{align*}
	Now, using $\gap=\sum_{A_H\setminus T} {\Deltac_a}^2/r_a$ and $r(C_a)\le \|r\|_1$ for all $a\in A_H\setminus T$ for the right term in the enumerator and $|\Deltac_a - \alpha_a r(C_a)|\le r(C_a)/2$ as the rounding error for the left term, yields
	\[
		\Exp \left[ \|\f'\|_R^2 - \|\f\|_R^2 \right]
		\le
		\frac{1}{4\tcn(T)} \|r\|_1 \tcn(T)
		- \frac{\gap}{\tcn(T)}
		\le
		\frac{1}{4\tcn(T)} \|r\|_1 m^2
		- \frac{\gap}{\tcn(T)}.
	\]
	Using $r_a\le 2s_a/x_a$, which follows from $\gamma\ge 2^6m^3\beta UC$, see~\eqref{form:multroundingerror} and the upper and lower bounds on $x,s$ in~\eqref{form:lowerupperbounds}, yields
	\[
		\|r\|_{1}m^2
		\le
		\frac{2(1+\delta)m^5\mu}{\Eps^2\beta^2}
		\le
		\frac{2(1+\delta)m^5\mu}{\Eps^2}\frac{\delta^2}{2^{10}m^6}
		\le
		\frac{(1+\delta)}{2\Eps^2}
		\frac{\delta^2\mu}{2^{8} m}
		\le
		\frac{\delta^2\mu}{2^{8}m}
		\le
		\gap,
	\]
	where we used $\beta\ge 2^5m^3/\delta$.
	This yields
	$
	\Exp \left[ \|\f'\|_R^2 - \|\f\|_R^2 \right]
	\le
	\tfrac{-3\gap}{4\tcn(T)}
	$.
\end{proof}
The above lemma yields a corresponding result to~\cite[Theorem 3.2]{DBLP:conf/stoc/KelnerOSZ13}, that is an expected guarantee on the approximation of the primal and dual energy. However, as Kelner et al.\ state ``one can use Markov bound or Chernoff bound to prove a probabilistic but exact guarantee''. In fact, the near-linear run-time of the \texttt{centering\_step} follows as in the proof of Theorem 2 in~\cite{DBLP:conf/isaac/BeckerK14}.

\section{Conclusion}\label{sec:conclusion}
We conclude by deducing the precise bound on the size of the numbers that the algorithm has to deal with: Note that the voltages, the values of the $\alpha$, as well as the currents $\f$ are each bounded in absolute value by $\|R\f\|_1$. By using the Cauchy-Schwarz-Inequality, we obtain $\|R\f\|_1\le \sqrt{\|r\|_1}\cdot\|\f\|_R$.
The first factor can be bounded by $\sqrt{2(1+\delta)\mu_0 m^3}/\eps\beta$ using the bounds from~\eqref{form:multroundingerror} and~\eqref{form:lowerupperbounds}.
The second, analogously to the estimation in~\eqref{form:phinormestimate}, can be bounded by $\sqrt{10\mu m}$. Together with $\mu\le \mu_0=3m\beta\gamma UC/\delta$, which is shown in Appendix~\ref{sec:initial} and using $\delta=1/8$, we obtain $2^8 m^{3} \gamma UC$ as the final bound on all numbers appearing in the centering step.
The values of $x$ and $s$ are bounded by $(1+\delta)\mu m/(\eps\gamma)$ and $(1+\delta)\mu m/(\eps\beta)$, respectively, which yields the bounds of $2^6m^2\beta UC$ and $2^6m^2\gamma UC$.
This shows Theorem~\ref{thm:proxynumbersize}. Finally, an arbitary given min-cost flow instance can first be scaled down, i.e., divided by its gcd, and then scaled up again with $\beta$ and $\gamma$ of size as described in Theorem~\ref{thm:proxynumbersize}. It follows that we can always guarantee the bound on the size of numbers in Theorem~\ref{thm:contribution}.

\newpage
\bibliography{refs}

\newpage
\appendix{
	\section{Appendix: Finding Arbitrarily Central Points}\label{sec:initial}
	In this section, we show how to construct an auxiliary instance and corresponding feasible interior solutions that are arbitrarily close to the central path. In fact, we can use the same technique as used in~\cite{DBLP:conf/isaac/BeckerK14} for a potential reduction method. Let us denote with $(G^0,b^0,u,c^0)$ the given input instance. Let $V^0$ denote the node set of $G^0$ and let $A^0$ denote its arc set. The method constructs an auxiliary instance $(G,b,c)$ that is an un-capacitated min-cost flow problem and corresponding primal and dual interior solutions for it. We denote with $V$ and $A$ the node and arc set of $G$, respectively, and with $n=|V|$ and $m=|A|$ their respective cardinalities.
	The first part of the construction is illustrated in Figure~\ref{fig:remcap} from left to middle. For each $a=(v,w)\in A^0$, insert the two nodes $v,w$ into $G$ as well as a new node $vw$ together with arcs $\acute{a}=(v,vw)$ and $\grave{a}=(w,vw)$, read ``\emph{$a$ up}'' and ``\emph{$a$ down}''. Define the new costs by $c_{\acute{a}} := c_a$ and $c_{\grave{a}} :=0$.
	The new demand vector is given by $b_v := b^0_v - u(\din(v))$ for nodes $v\in V^0$ and $b_{vw}:=u_{(v,w)}$ for the newly inserted nodes. A flow of $f$ on $a\in G^0$ corresponds to a flow of $f$ on $\acute a=(v,vw)$ and a flow of $u_a-f$ on $\grave a = (w,vw)$ in $G$. This is a well-known construction and it is known that the new instance is equivalent to the old instance, see for example~\cite[Section 2.4]{DBLP:books/daglib/0069809}.
	\begin{figure}[ht!]
		\begin{center}
			\resizebox{!}{3.8cm}{
				\input{remcap}
			}
		\end{center}
		\caption{
			The transition from the left to the middle, which is done for each arc, removes the
			capacity constraint. From the middle to the right: In order to balance the $x_as_a$,
			we introduce the arc $\hat a = (v,w)$ with high cost $c_{\hat a}$ and reroute flow
			along it. The direction of $\hat a$ depends on a tree solution $z$ in the original graph. It is flipped if $z_a \le u_a / 2$. All arcs in the middle and the right graph have infinite capacity. Only the costs are shown.}
		\label{fig:remcap}
	\end{figure}
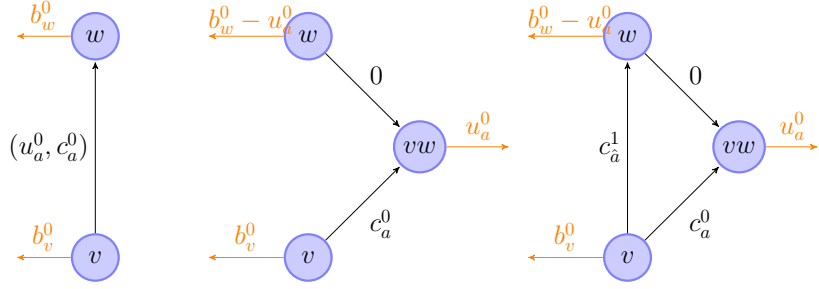

	The goal is to find values $x,y,s$ that are primal and dual feasible, integral, as well as close to the central path.
	The latter means that we want that $\sum_{a\in A}|\tfrac{x_as_a}{\mu_0}-1|\le \delta$ for some parameters $\mu_0>0$ and $\delta\ge 0$. We show that the technique from~\cite{DBLP:conf/isaac/BeckerK14} can actually achieve this. With a different choice of parameters, the technique yields closeness to the central path for arbitrary $\delta$.

	The second part of the construction starts by computing an integral (not necessarily feasible) tree solution $z$ in $G^0$ for an arbitrary spanning tree $T$, set $x_{\grave a}=x_{\acute a}=u_a/2$ and introduce the additional arc $\hat a =\sign(z_a - u_a/2) (v,w)$. If $z_a-u_a/2=0$, do not introduce the arc $\hat a$. Set the flow on $\hat a$ to $x_{\hat a}=|z_a - u_a/2|$ and the cost to $c_{\hat a}=\lceil t/|z_a-u_a/2|\rceil$.
	A sufficiently large choice of $t$ guarantees that no flow is routed across $\hat a$ in an optimal solution. More precisely, for $t\ge 2\beta \gamma m UC$, we obtain $c_{\hat a}\ge t/(2\beta U)\ge m \gamma C\ge \ones^Tc^0$, where we used that $z_a\le \|b^0\|_1/2$. It follows that the cost of $\hat a$ is at least the cost of any path in $G^0$ and thus introducing $\hat a$ does not change the optimum of the problem.
	The dual variables are set as $y_v,y_w = 0$ and $y_{vw} = -2t/u_a$, which determines $s_{\acute a} = c_a +2t/u_a$, $s_{\grave a} = 2t/u_a$ and $s_{\hat a} = c_{\hat a} = \lceil t/|z_a-u_a/2|\rceil$.
	We can now show that $x,y,s$ are arbitrarily close to the central path.
	\begin{lemma}\label{lem:init}
		Let $x,y,s$ as described above.
		\begin{enumerate}
			\item For any $t> 0$, it holds that $x_a s_a \in [t,t+ \beta\gamma UC]$ for all $ a\in A$.
			\item Let $t\ge \beta\gamma UC (\tfrac{m}{\delta}-1)$ and $\mu_0:=t+\beta\gamma UC$, then $\|\sigma\|_1 \le \delta$, where $\sigma_a = \tfrac{x_as_a}{\mu_0}-1$.
		\end{enumerate}
	\end{lemma}
	\begin{proof}
		\begin{enumerate}
			\item Let $a\in A_0$ be any arc in $G_0$, then $x_{\vvw a} = x_{\wvw a} = u_a/2$. It follows that,
			      \begin{align*}
			      	  & x_{\vvw a} s_{\vvw a} = \frac{u_a}{2} \left(c_{\vvw a} + \frac{2t}{u_a}\right) = t + \frac{ u_a c_a}{2} \le t + \frac{\beta\gamma UC}{2},
			      	\quad
			      	x_{\wvw a} s_{\wvw a} = \frac{u_a}{2} \left(c_{\wvw a} + \frac{2t}{u_a} \right)=t
			      	\quad \text{and}\quad\\
			      	  & x_{\hat a} s_{\hat a}
			      	\ge
			      	\Big|z_a - \frac{u_a}{2}\Big| \frac{t}{|z_a - \frac{u_a}{2}|} = t
			      	\quad \text{and} \quad
			      	x_{\hat a} s_{\hat a}
			      	\le
			      	\Big|z_a - \frac{u_a}{2}\Big| \left(\frac{t}{|z_a - \frac{u_a}{2}|} + 1 \right)
			      	\le
			      	t + \beta U.
			      \end{align*}
			\item
			      From the first part, it follows that $x_a s_a \in [t,t+ \beta\gamma UC]$ for $a\in A$ or equivalently $\frac{x_as_a}{\mu_0}\in[t/(t+\beta\gamma UC),1]=[1 - \beta\gamma UC/(t+\beta\gamma UC),1]$. This yields $\frac{x_as_a}{\mu_0}\ge 1 - \delta/m$ with the lower bound on $t$ and hence
			      \[
			      	\|\sigma\|_1
			      	=
			      	\sum_{a\in A} \left|\frac{x_as_a}{\mu_0}-1\right|
			      	\le
			      	\sum_{a\in A}\frac{\delta}{m}=\delta. \tag*{\qedhere}
			      \]
		\end{enumerate}
	\end{proof}
	From the above lemma and the lower bound on $t$ that ensures that the optimum is not changed due to the introduction of the new arcs $\hat a$, we obtain that choosing $\mu_0=3m_0\beta\gamma UC/\delta$ is sufficient. Choosing $\delta=1/8$ implies that setting $\mu_0=24m\beta\gamma UC$ is sufficient.
}
\end{document}

%% file: macros.tex
\newcommand{\p}{\pi}

\newcommand{\f}{\varphi}

\newcommand{\supd}{\Delta s}
\newcommand{\xupd}{\Delta x}
\newcommand{\yupd}{\Delta y}

\DeclareMathOperator{\tcn}{tcn}

\newcommand{\vvw}{\acute}
\newcommand{\wvw}{\grave}

\DeclareMathOperator{\sign}{sign}
\DeclareMathOperator{\diag}{diag}

\DeclareMathOperator{\gap}{gap}

\DeclareMathOperator{\Exp}{E}

\let\eps\varepsilon

\newcommand{\ones}{\mathds{1}}

\newcommand{\NN}{\ensuremath{\mathbb{N}}}

\newcommand{\ZZ}{\ensuremath{\mathbb{Z}}}

\newcommand{\RR}{\ensuremath{\mathbb{R}}}

\DeclareMathOperator{\Iin}{in}
\newcommand{\din}{\ensuremath{\delta^{\Iin}}}

\DeclareMathOperator{\Oout}{out}
\newcommand{\dout}{\ensuremath{\delta^{\Oout}}}



%% file: remcap.tex
\begin{tikzpicture}[shorten >=1pt,node distance=2cm,>=stealth',initial/.style={}]
\begin{scope}
\tikzstyle{every state}=[draw=blue!50,very thick,fill=blue!20]
  \node[state, color = white]			(vw) 			{\Large{$vw$}};
  \node[state]          	(w) [above left =of vw]	{\Large{$w$}};
  \node[state]          	(v) [below left =of vw]	{\Large{$v$}};
\tikzset{mystyle/.style={orange}}

\path[->]	(v) edge 		node[left] {\Large{$(u^0_a,c^0_a)$}} (w);
\path[->]		(w) edge[mystyle] 	node[above] {\Large{$b^0_w$}}  ++ (-1.5,0);
\path[->]		(v) edge[mystyle] 	node[above] {\Large{$b^0_v$}}  ++ (-1.5,0);
\end{scope}

\begin{scope}[xshift=4cm]
\tikzstyle{every state}=[draw=blue!50,very thick,fill=blue!20]
  \node[state]			(vw) 			{\Large{$vw$}};
  \node[state,color=white]	(a) [left =of vw]	{};
  \node[state]          	(w) [above left =of vw]	{\Large{$w$}};
  \node[state]          	(v) [below left =of vw]	{\Large{$v$}};
\tikzset{mystyle/.style={orange}}
\path[->]	(v) edge 		node[below right] {\Large{$c^0_a$}} (vw);
\path[->]	(w) edge 		node[above right] {\Large{$0$}} (vw);
\path[->]		(v) edge[mystyle] 	node[above] {\Large{$b^0_v$}}  ++ (-1.9,0);
\path[->]		(w) edge[mystyle] 	node[above] {}  ++ (-1.9,0);
\node[color=orange] (d) at (-3.17,2.4) {\Large{$b^0_w - u^0_a$}};
\path[->]		(vw) edge[mystyle] 	node[above] {\Large{$u^0_{a}$}}  ++ (1.7,0);
\end{scope}

\begin{scope}[xshift=10 cm]
\tikzstyle{every state}=[draw=blue!50,very thick,fill=blue!20]
  \node[state]			(vw) 			{\Large{$vw$}};
  \node[state]          	(w) [above left =of vw]	{\Large{$w$}};
  \node[state]          	(v) [below left =of vw]	{\Large{$v$}};
\tikzset{mystyle/.style={orange}}
\path[->]	(v) edge 		node[below right] {\Large{$c^0_{a}$}} (vw);
\path[->]	(w) edge 		node[above right] {\Large{$0$}} (vw);
\path[->]	(v) edge 		node[left] {\Large{$c^1_{\hat a}$}} (w);
\path[->]		(w) edge[mystyle] 	node[above] {}  ++ (-1.9,0);
\node[color=orange] (d) at (-3.17,2.4) {\Large{$b^0_w - u^0_a$}};
\path[->]		(v) edge[mystyle] 	node[above] {\Large{$b^0_v$}}  ++ (-1.9,0);
\path[->]		(vw) edge[mystyle] 	node[above] {\Large{$u^0_a$}}  ++ (1.5,0);
\end{scope}
\end{tikzpicture}